\documentclass[aps,groupedaddress,superscriptaddress,amssymb,amsmath,amsbsy,amsfonts,twocolumn,pra]{revtex4}

\usepackage{amssymb,amsmath,amsthm,color,graphicx,times,graphicx}
\usepackage{hyperref}
\usepackage{enumerate}
\usepackage{subfigure}
\usepackage{dsfont}
\usepackage{times}
\usepackage{txfonts}
\usepackage{bbold}

\newcommand*{\cl}[1]{{\mathcal{#1}}}
\newcommand*{\bb}[1]{{\mathbb{#1}}}
\newcommand{\ket}[1]{\left|#1\right>}

\newcommand{\proj}[2]{| #1 \rangle\!\langle #2 |}
\newcommand*{\tn}[1]{{\textnormal{#1}}}
\newcommand*{\1}{{\mathbb{1}}}
\newcommand{\T}{\mbox{$\textnormal{Tr}$}}

\theoremstyle{plain}

\newtheorem{proposition}{Proposition}
\newtheorem{theorem}{Theorem}
\newtheorem{lemma}[theorem]{Lemma}

\theoremstyle{definition}
\newtheorem{definition}{Definition}

\date{\today}

\begin{document}

\title{
Approximate private quantum channels on fermionic Gaussian systems
}

\author{Kabgyun Jeong}
\email{kgjeong6@snu.ac.kr}
\affiliation{Research Institute of Mathematics, Department of Mathematical Sciences, Seoul National University, Seoul 08826, Korea}
\affiliation{School of Computational Sciences, Korea Institute for Advanced Study, Seoul 02455, Korea}

\begin{abstract}
The private quantum channel (PQC) maps any quantum state to the maximally mixed state for the discrete as well as the bosonic Gaussian quantum systems, and it has fundamental meaning on the quantum cryptographic tasks and the quantum channel capacity problems. In this paper, we introduce a notion of approximate private quantum channel ($\varepsilon$-PQC) on \emph{fermionic} Gaussian systems (i.e., $\varepsilon$-FPQC), and construct its explicit form of the fermionic (Gaussian) private quantum channel. First of all, we suggest a general structure for $\varepsilon$-FPQC on the fermionic Gaussian systems with respect to the Schatten $p$-norm class, and then we give an explicit proof of the statement in the trace norm. In addition, we study that the cardinality of a set of fermionic unitary operators agrees on the $\varepsilon$-FPQC condition in the trace norm case.
\end{abstract}

\maketitle

\section{Introduction}
In general, we can classify two intrinsic physical systems known as a bosonic system and a fermionic system. Each physical systems undergo a certain unitary transformation known as a state evolution or a quantum channel, which is mathematically completely positive and trace-preserving (CPT) map in quantum information theory, from a quantum state to other one~\cite{H16,W17,W18}. Beside the bosonic quantum states and (bosonic) channels are familiar in quantum information theory~\cite{WPG+12},  the fermionic systems and its informational properties are relatively unknown~\cite{B05,SSGK18}. In this reason, we try to investigate a fermionic quantum channel in the Gaussian regime, and then construct a (Gaussian) fermionic private quantum channel (FPQC). The notion of private quantum channel is very useful in the quantum cryptographic protocols as well as the channel capacity problems in QIT. For example, two conjugate pairs of private quantum channels give rise to an additivity violation of the classical capacity for quantum channels~\cite{H09,HW08}, so we naturally expect that FPQC also has the non-additive property.

The private quantum channel (PQC), first introduced by Ambainis \emph{et al}.~\cite{AMTW00}, is a quantum communication primitive for secure transmission of a quantum information, and already has been proved not only the informational-security including its optimality~\cite{NK06,BZ07} but also reported several asymptotic secure transmission rates~\cite{CWY04,D05,H15}. After applying the PQC on any quantum state, the output of the channel is always equivalent to the maximally mixed state, which has most highest entropy for a given dimension of the state, thus any wiretappings are fundamentally impossible. Owing to cryptographic importance of PQC, it has several names such as quantum one-time pad, random unitary channel, $\varepsilon$-randomizing map and so on, here we call the map as $\varepsilon$-private quantum channel ($\varepsilon$-PQC) in the approximate consideration. While the conventional PQC is required exactly $d^2$ unitary operations to encrypt a $d$-dimension quantum state to the perfect maximally mixed state, $\varepsilon$-PQC (i.e., approximate PQC) is only sufficient to use the number of unitary operations being less than $O(d\log d)$~\cite{HLSW04}. 

Here, let us formally define the $\varepsilon$-PQC in the general-setting through the Schatten $p$-norm class~\cite{J14}: For every quantum state $\varrho\in\cl{B}(\bb{C}^d)$ and any $\varepsilon>0$, if a quantum channel $\Lambda:\cl{B}(\bb{C}^{d})\to\cl{B}(\bb{C}^{d})$ satisfies the following inequality of
\begin{equation}
\left\|\Lambda(\varrho)-\frac{\1}{d}\right\|_p\le\frac{\varepsilon}{\sqrt[p]{d^{p-1}}},
\end{equation}
then we \emph{call} the map $\Lambda$ as $\varepsilon$-\emph{private quantum channel} with respect to the Schatten $p$-norm (for all $p\ge1$)~\cite{AS04,DN06}. Notice that $\cl{B}(\bb{C}^{d})$ denotes the bounded linear operator on the $d$-dimensional Hilbert space $\bb{C}^d$, and the Schatten $p$-norm is defined as follows: For any matrix $A\in\cl{B}(\bb{C}^d)$ and for all $1\le p\le\infty$, it has in the form of trace class as
\begin{equation*}
\|A\|_p=\left[\T(A^\dag A)^{p/2}\right]^{1/p}.
\end{equation*}
For convenience, we only consider $p=1$ case below, i.e., the trace norm given by $\|A\|_1=\T\sqrt{A^\dag A}$, however, we formulate the fermionic $\varepsilon$-PQC for arbitrary $p\ge1$ (see Proposition~\ref{prop1}). The operator norm and the Hilbert-Schmidt norm are given similar ways~\cite{AS04,HLSW04}. Thus the $\varepsilon$-PQC in this trace class is taken in the form of $\left\|\Lambda(\varrho)-\frac{\1}{d}\right\|_1\le{\varepsilon}$. Also, there are several variants of the PQC in continuous-variable regimes~\cite{Br05,JKL15,WCHZ14} and the multi-qubit protocol~\cite{JK15}.

Now it is natural to ask how we can characterize the PQC in fermionic Gaussian systems and their impact on channel-capacity problems. At first, we briefly review the basic concepts of fermionic Gaussian systems and the channels.

This paper is organized as follows. In Subsecs.~\ref{subsec:fermion} and \ref{subsec:fpqc}, we review the basic of fermionic Gaussian systems and its representation of private quantum channels, respectively. In Sec.~\ref{sec:main}, we describe our main result on approximate private quantum channels on the fermionic system with explicit construction and proof over the trace norm. Finally we conclude our result in Sec.~\ref{sec:conclusion}.

\subsection{Fermionic Gaussian states} \label{subsec:fermion}
Generally, $M$-mode fermionic systems are associated with a tensor product of  Hilbert space $\cl{H}^{\otimes M}=\bigotimes_{j=1}^M\cl{H}_j$, where $M$ pairs of fermionic annihilation and creation operators $\{\hat{f}_j,\hat{f}_j^\dag\}_{j=1}^M$ correspond to each mode of the total Hilbert space. The operators in the form of $\hat{\mathbf{f}}^T:=(\hat{f}_1,\ldots,\hat{f}_M,\hat{f}_1^\dag,\ldots,\hat{f}_M^\dag)$ satisfy the canonical anti-commutation relation (CAR) such that $\{\hat{f}_k,\hat{f}_\ell^\dag\}=\hat{f}_k\hat{f}_\ell^\dag+\hat{f}_\ell^\dag\hat{f}_k=\delta_{k\ell}\1$. It was known that CAR algebra of $M$-mode fermionic system can be described by a set of the $M$-mode Majorana operators $\{\hat{c}_1,\ldots,\hat{c}_{2M}\}$ such that $\{\hat{c}_k,\hat{c}_{\ell}\}=2\delta_{k\ell}\1$ as well as $\hat{c}_k=\hat{c}_k^\dag$, and those operators in the Clifford algebra have an explicit form:
\begin{align}
\left\{ \begin{array}{ll} 
\hat{c}_{2k-1}&=\sigma_1^z\otimes\cdots\otimes\sigma_{k-1}^z\otimes\sigma_k^x\otimes\1_2\otimes\cdots\otimes\1_2 \\ & \\
\hat{c}_{2k}&=\sigma_1^z\otimes\cdots\otimes\sigma_{k-1}^z\otimes\sigma_k^y\otimes\1_2\otimes\cdots\otimes\1_2,
\end{array}\right.
\end{align}
where $\1_2=\left(\begin{array}{cc} 1 & 0 \\ 0 & 1 \end{array}\right), \sigma_k^x=\left(\begin{array}{cc} 0 &1 \\ 1 & 0 \end{array}\right), \sigma_k^y=\left(\begin{array}{cc} 0 &i \\ -i & 0 \end{array}\right)$, and $\sigma_k^z=\left(\begin{array}{cc} 1 &0 \\ 0 & -1 \end{array}\right)$ are Pauli matrices on the $k$-th qubit. Note that, for each $k$-mode, $\hat{f}_k\hat{f}_k^\dag=\proj{0}{0}_k=\frac{1}{2}(1-i\hat{c}_{2k-1}\hat{c}_{2k})$, $\hat{f}_k^\dag\hat{f}_k=\proj{1}{1}_k=\frac{1}{2}(1+i\hat{c}_{2k-1}\hat{c}_{2k})$, and $\hat{\mathbf{c}}^T:=(\hat{c}_1,\ldots,\hat{c}_{2M})$.

\begin{definition}[Fermionic Gaussian state] 
A fermionic state $\rho_F$ is Gaussian, if it can be defined by
\begin{equation}
\rho_F=\lim_{\beta\to\pm\infty}\frac{e^{\beta\hat{H}}}{Z(\beta)},
\end{equation}
where $\beta$ is the inverse temperature, $Z(\beta)=\T(e^{\beta\hat{H}})$ the normalization factor, and the second order Hamiltonian $\hat{H}$ is given by 
\begin{equation}
\hat{H}=\frac{i}{2}\hat{\mathbf{c}}^T\Gamma\hat{\mathbf{c}}+\hat{\mathbf{c}}^T\mathbf{x}.
\end{equation}
Here, $\Gamma=-\Gamma^T\in \cl{M}_{2M}(\bb{R})$ is a real skew-symmetric matrix and $\mathbf{x}\in\bb{R}^{2M}$. For convenience, we will set the temperature parameter as $\beta=1$.
\end{definition}
Now, we only consider the quadratic term of the Hamiltonian $\hat{H}'=\frac{i}{2}\hat{\mathbf{c}}^T\Gamma\hat{\mathbf{c}}$, i.e., fermionic ``even" Gaussian states. For $M$-mode fermionic cases, a Gaussian unitary is naturally given by $e^{i\hat{H}}\in\cl{U}(2M)$, which can be decomposed in the form of $e^{i(\hat{H}_1+\hat{H}_2)}$ through the Lie theory. Then it was known that there \emph{exist} a fermionic Gaussian unitary and $2M\times2M$ orthogonal matrix $e^{\Gamma}\in\tn{SO}(2M)$ satisfying
\begin{equation}
e^{i\hat{H}}\hat{\mathbf{c}}e^{-i\hat{H}}=e^{\Gamma}\hat{\mathbf{c}}.
\end{equation}

For any ($M$-mode fermionic) \emph{even} Gaussian states $\rho_F$, note that there exists a Gaussian unitary operator $e^{i\hat{H}'}$ and an orthogonal matrix $O_\Gamma\in\tn{SO}(2M)$ such that
\begin{align}
\rho_F&=\frac{1}{Z}e^{\hat{H}'} \nonumber\\
&=e^{i\hat{H}'}\cdot\frac{1}{Z}e^{\frac{i}{2}\hat{\mathbf{c}}^TO_\Gamma AO_{\Gamma}^T\hat{\mathbf{c}}}\cdot e^{-i\hat{H}'} \nonumber\\
&=\bigotimes_{k=1}^M\left(\frac{1+\lambda_k}{2}\proj{0}{0}_k+\frac{1-\lambda_k}{2}\proj{1}{1}_k\right)\equiv\bigotimes_{k=1}^M\rho_{F,k},
\end{align}
where $A=\bigoplus_{k=1}^M\left(\begin{array}{cc} 0 & \lambda_j \\ -\lambda_j & 0\end{array}\right)$ with its spectrum $\lambda_k\in[0,1]$.
Furthermore, for any $k$, if $\lambda_k=1$, then $\rho_F\in\cl{H}(\bb{C}^{2M})$ is said to be a pure state, i.e., pure fermionic Gaussian state is given by $\rho_F=\proj{0}{0}_1\otimes\cdots\otimes\proj{0}{0}_M$. In those cases, the entropy of the $M$-mode fermionic Gaussian state is defined by
\begin{equation}
S(\rho_F)=\sum_{k=1}^MS(\rho_{F,k}),
\end{equation}
where $S(\rho_{F,k}):=-\frac{1+\lambda_k}{2}\log\frac{1+\lambda_k}{2}-\frac{1-\lambda_k}{2}\log\frac{1-\lambda_k}{2}$ is the von Neumann entropy of the fermionic system.

It is very useful to study a private quantum channel in quantum information science, because an output of PQC gives birth to a maximally mixed state (MMS) at the end of the channel. This output state of the channel directly corresponds to a maximally entangled state (MES) as in Ref.~\cite{JL16,LKLJ19} via a quantum purification method~\cite{HJW93}. However, still a notion of fermionic private quantum channel does not exist.

\subsection{Representation of fermionic Gaussian quantum channels: The fermionic private quantum channel} \label{subsec:fpqc}
In 2005, Bravyi first introduced the notion of fermionic Gaussian quantum channels as follows~\cite{B05}: For any completely positive and trace-preserving map, the fermionic Gaussian channel $\Lambda_F$ is given by
\begin{align}
\Lambda_F(\hat{c}_k)&=\xi_k\hat{c}_k,\;\;\;\forall k=1,\ldots,2M\;\;\;\tn{and} \nonumber\\
\Lambda_F(\hat{\mathbf{c}}(\vec{b}))&=\prod_{k;b_k=1}\xi_k\hat{\mathbf{c}}(\vec{b}),
\end{align}
where $\hat{\mathbf{c}}(\vec{b})=\hat{c}_1^{b_1}\hat{c}_2^{b_2}\cdots\hat{c}_{2M}^{b_{2M}}$ with a binary string $\vec{b}=(b_1,\ldots,b_{2M})$. Here we note that $0\le \xi_1,\ldots,\xi_{2M}\le1$ are real parameters characterizing the fermionic quantum channels and it is called the \emph{attenuation} coefficient. Now, we are ready to define the fermionic private quantum channel in the form of Kraus representation. For any fermionic Gaussian state $\rho_F$, the fermionic Gaussian channel $\Lambda_F$ is described by
\begin{eqnarray}
\Lambda_F(\rho_F)&=\frac{1}{|\cl{U}|}\sum_{\ell=1}^{|\cl{U}|}U_\ell\rho_FU_\ell^\dag, 
\end{eqnarray}
where $U_\ell=i\pi\hat{c}_\ell$ such that $U_\ell\hat{c}_m=(-1)^{\delta_{\ell m}}\hat{c}_mU_\ell$ with $\pi=(-1)^M\hat{c}_1\cdots\hat{c}_{2m}$. Note that $|\cl{U}|$ denotes the cardinality of the unitaries on the unitary group $\cl{U}$. In the optimal case, the cardinality of $\cl{U}$ is given by $|\cl{U}|=(2M)^2$. (See Fig.~\ref{Fig1}.)

\begin{figure}[t!]
\center
\includegraphics[width=6.5cm]{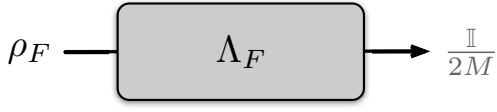}
    \caption{Schematic diagram for $M$-mode fermionic private quantum channel $\Lambda_F$. For any fermionic Gaussian state $\rho_F$, if an output of the channel is $\frac{\1}{2M}$, then we call the channel as \emph{perfact} fermionic PQC. Otherwise, i.e., the channel's output is almost close to $\frac{\1}{2M}$, then the channel is called as fermionic $\varepsilon$-PQC or $\varepsilon$-FPQC.}
\label{Fig1}
\end{figure}

\begin{definition}[Fermionic $\varepsilon$-private quantum channel]
For any fermionic (Gaussian) state $\rho_F$ and any $\varepsilon>0$, if a fermionic Gaussian quantum channel $\Lambda_F:\cl{H}^{\otimes M}\to\cl{H}^{\otimes M}$ satisfies
\begin{equation}
\left\|\Lambda_F(\rho_F)-\frac{\1}{2M}\right\|_p\le\frac{\varepsilon}{\sqrt[p]{d^{p-1}}},
\end{equation}
then the channel $\Lambda_F(\cdot)$ is said to be fermionic $\varepsilon$-\emph{private quantum channel} (or $\varepsilon$-FPQC) with respect to the Schatten $p$-norm (for all $p\ge1$).
\end{definition}

For the case of $p=1$, $\varepsilon$-FPQCs are taken in the form of $\left\|\Lambda_F(\rho_F)-\frac{\1}{2M}\right\|_1\le{\varepsilon}$.
Then, we are ready to state our main result of the approximate fermionic private quantum channel for randomizing fermionic Gaussian quantum states.

\section{Main results} \label{sec:main}
We have briefly reviewed the concrete notions on fermionic Gaussian systems and the definition of the approximate private quantum channels on fermionic systems, so now we will introduce our main results. While the proposed results are simple, the proofs are subtle complicated. However, the statements have a fundamental meaning in quantum communication theory on whether the fermionic Gaussian systems are tractable in quantum channel capacity problems or not. If we find an explicit form of the fermionic PQC similar to the bosonic PQCs, e.g., in Refs.~\cite{Br05,JKL15}, we can further argue on the topic of the quantum channel capacity problems as well as its non-additive properties.

According to the Hayden \emph{et al.}'s~\cite{HLSW04}, Dickinson and Nayak's~\cite{DN06}, and also the Author's previous result~\cite{J14,J19}, we suggest a following proposition.

\begin{proposition}[Approximate fermionic PQC] Let $\rho_F$ be an $M$-mode fermionic Gaussian state, and $\Lambda_F(\rho_F)=\frac{1}{|\cl{U}|}\sum_{\ell=1}^{|\cl{U}|}U_\ell\rho_F U_\ell^\dag$ be an $\varepsilon$-FPQC with respect to the Schatten $p$-norm. Then, for any $\varepsilon>0$ and for sufficiently large $M$, there exists a set of fermionic unitary operators $\{U_\ell=i\pi\hat{c}_\ell\}_{\ell=1}^{|\cl{U}|}$ in $\cl{U}(2M)$ with the cardinality at least
\begin{equation}
|\cl{U}|\ge2\kappa M\log\frac{10(2M)^{(p-1)/p}}{\varepsilon},
\end{equation}
where $\kappa$ is an absolute constant.
\label{prop1}
\end{proposition}


Here, we present a proof only on the case of $p=1$ as mentioned in Introduction. In this case, we can fix the logarithmic factor as $\log\frac{10}{\varepsilon}$, and from the independence of the mode $M$, it can be omitted. Also, notice that the cardinality could be reduced by $4M^2$ to $\cl{O}(2M\log2M)$ by the proposition~\ref{prop1}.

\begin{proposition}[$\varepsilon$-FPQC for $p=1$ case]
Let $\Lambda_F(\rho_F)=\frac{1}{|\cl{U}|}\sum_{\ell=1}^{|\cl{U}|}\pi\hat{c}_\ell\rho_F (\pi\hat{c}_\ell)^\dag$ be an $\varepsilon$-FPQC with respect to the trace norm. Then, for any $\varepsilon>0$ and $M\gg1$, there exists a set of Majorana operators $\{i\pi\hat{c}_\ell\}_{\ell=1}^{|\cl{U}|}$ in $\cl{U}(2M)$ with the cardinality of
\begin{equation}
|\cl{U}|\ge2\kappa M,
\end{equation}
where $\kappa$ is also an absolute constant as in Proposition~\ref{prop1}.
\label{prop2}
\end{proposition}

For the proof we are required two technical lemmas. Below Lemma~\ref{net} states that pure quantum states on Bloch sphere on any dimension can be discretized into a net-point on the regularized polyhedron in the dimension, and Lemma~\ref{McD} endows us to estimate an exponentially decaying of the tale probability distribution on a random variable. Those lemmas are universal not only in the bosonic Gaussian system but also in the fermionic Gaussian one.

\begin{lemma}[$\varepsilon$-net~\cite{HLSW04}] \label{net}
Let $\varepsilon>0$ and the Majorana mode $M\gg1$. For any fermionic pure Gaussian states $\ket{\varphi_F}\in\cl{H}^M$, we can choose a net point $\ket{\tilde{\varphi}_F}\in N$ such that $\|\varphi_F-\tilde{\varphi}_F\|_1\le\varepsilon$. Then there exists a net $N$ of pure fermionic states satisfying
\begin{equation}
\|N\|\le\left(5/\varepsilon\right)^{4M}.
\end{equation}
\end{lemma}

\begin{lemma}[McDiarmid inequality~\cite{M89}]
Let $\{X_k\}_{k=1}^m\subset{\cl{S}}$ be independent random variables chosen uniformly at random. Let a measurable function $F:\cl{S}^m\to\bb{R}$ satisfy $|F(x)-F(\hat{x})|\le c_k$, called the bounded difference, where the vectors $x$ and $\hat{x}$ differ only in the $k$-th position. If we define a random variable $Y=F(X_1,\ldots,X_m)$, then ($\forall t\ge0$)
\begin{equation}
\tn{Pr}[|Y-\mathbf{E}(Y)|\ge t]\le2e^{-2t^2/\sum_{k=1}^mc_k^2},
\end{equation}
where $\mathbf{E}(Y)$ denotes the expectation value for the random variable $Y$.
\label{McD}
\end{lemma}

In fermionic Gaussian regimes, suppose that the fermionic PQC $\Lambda_F$ is realized by a sequence of the Majorana operators $(i\pi\hat{c}_k)_{k=1}^{|\cl{U}|}$, and the other map ${\Lambda}_F'$ is given by $(i\pi\hat{c}_1,\ldots,i\pi\hat{c}_k',\ldots,i\pi\hat{c}_{|\cl{U}|})$, respectively. Then we have the \emph{bounded} difference as
\begin{align*}
&\left|\left\|\Lambda_F(\varphi_F)-\frac{\1}{2M}\right\|_1-\left\|\Lambda_F'(\varphi_F)-\frac{\1}{2M}\right\|_1\right| \\
&~~~~~~~~~~~~~~~~~~~~~~~~~~~~~~\le\left\|\Lambda_F(\varphi_F)-\Lambda_F(\varphi_F)'\right\|_1 \\
&~~~~~~~~~~~~~~~~~~~~~~~~~~~~~~=\frac{1}{|\cl{U}|}\left\|\pi\hat{c}_k\varphi_F(\pi\hat{c}_k)^\dag-\pi\hat{c}_k'\varphi_F(\pi\hat{c}_k')^{\dag}\right\|_1 \\
&~~~~~~~~~~~~~~~~~~~~~~~~~~~~~~\le\frac{2}{|\cl{U}|},
\end{align*}
where we make use of the norm convexity and the fact of $\|\phi-\phi'\|_1\le2$ for any quantum states. From the McDiarmid inequality (on the positive part), we estimate that
\begin{equation}
\tn{Pr}\left[Y_{\varphi_F}\ge t+\left(\frac{2M}{|\cl{U}|}+\frac{1}{2M}\right)\right]\le e^{-|\cl{U}|t^2/2},
\end{equation}
where $Y_{\varphi_F}:=\|\Lambda_F(\varphi_F)-\frac{\1}{2M}\|_1$.

Now, we are ready to prove the Proposition~\ref{prop2}. That is, $\varepsilon$-FPQC is fulfilled when we take the fermionic unitary operators as in the order of the cardinality $|\cl{U}|$.

\begin{proof}
Let the set of Majorana operators $\{i\pi\hat{c}_{k}\}_{k=1}^{|\cl{U}|}$ be an i.i.d. random variable distributed according to the Haar measure. We can prove that the fernionic map $\Lambda_F$ is the $\varepsilon$-FPQC in high probability.

If we fix the net $N$ in Lemma~\ref{net}, and define $\tilde{\varphi}_F$ to be a net point on the fermionic pure Gaussian states $\varphi_F$. Then, by the unitary invariance, we can conclude that
\begin{equation} \label{norm inv}
\|\Lambda_F(\varphi_F)-\Lambda_F(\tilde{\varphi}_F)\|_1=\|\varphi_F-\tilde{\varphi}_F\|_1\le\frac{\varepsilon}{2}.
\end{equation}
Thus, from the $\varepsilon$-net lemma, we can obtain the net with the cardinality $|N|\le(20M/\varepsilon)^{4M}$. This implies that
\begin{align}
&\tn{Pr}_{\forall\varphi_F}\left[\left\|\Lambda_F(\varphi_F)-\frac{\1}{2M}\right\|_1\ge\varepsilon\right] \nonumber\\
&\le\tn{Pr}_{\forall\varphi_F,\tilde{\varphi}_F}\left[\|\Lambda_F(\varphi_F)-\Lambda_F(\tilde{\varphi}_F)\|_1+\left\|\Lambda_F(\tilde{\varphi_F})-\frac{\1}{2M}\right\|_1\ge\varepsilon\right] \\
&\le\tn{Pr}_{\forall\tilde{\varphi}_F}\left[\left\|\Lambda_F(\tilde{\varphi}_F)-\frac{\1}{2M}\right\|_1\ge\frac{\varepsilon}{2}\right],\;\;\;(*) \nonumber
\end{align}
where we have used the triangle inequality and Eq.~(\ref{norm inv}).

Finally, by using the union bound and the net construction above, we can derive following inequalities:
\begin{align}
(*)&\le|N|\cdot\tn{Pr}_{\forall_{\tilde{\varphi}_F^{(1)}}}\left[\left\|\Lambda_F(\tilde{\varphi}_F^{(1)})-\frac{\1}{2M}\right\|_1\ge\frac{\varepsilon}{2}\right] \nonumber\\
&\le 2\left(\frac{20M}{\varepsilon}\right)^{4M}\exp\left[-|\cl{U}|\left(\frac{\varepsilon}{4M}-\frac{(2M)^{1/2M}}{|\cl{U}|}-\frac{1}{2M}\right)^2\right].
\label{finalprob}
\end{align}
This completes the proof if the probability is bounded by 1 (see Lemma~\ref{upper bound} below), and $|\cl{U}|\ge2\kappa M$ with $\kappa:=\frac{1}{c\varepsilon^2}\log\left(\frac{10}{\varepsilon}\right)$ for a constant $c$.
\end{proof}

\begin{lemma}\label{upper bound}
For sufficiently large $M$, if we take the cardinality as in the form of
\begin{equation}
|\cl{U}|\ge2M\frac{1}{c\varepsilon^2}\log\left(\frac{10}{\varepsilon}\right),
\end{equation}
then the probability we required in Eq.~(\ref{finalprob}) is upper bounded by 1.
\end{lemma}

\begin{proof}
For sufficiently large $|\cl{U}|$ satisfying $2M<|\cl{U}|<(2M)^2$, we can take the bound as
\begin{equation}
2\left(\frac{20M}{\varepsilon}\right)^{4M}\exp\left[-|\cl{U}|\left(\frac{\varepsilon}{4M}-\frac{(2M)^{1/2M}}{|\cl{U}|}-\frac{1}{2M}\right)^2\right]<1.
\end{equation}
By the straightforward calculation, this bound gives rise to $|\cl{U}|\le4M^2$ we expected. Here, if we fix the mode $M$ and choose $2M$ so that $\left(\frac{\varepsilon}{2}-\frac{2(2M)^{1/2M}}{|\cl{U}|}\right)^2=o(\varepsilon^2)$, then we have
\begin{equation*}
2M\log\left(\frac{10}{\varepsilon}\right)<c|\cl{U}|\varepsilon^2.
\end{equation*}
\end{proof}

This construction shown to be that it is possible to construct an approximate fermionic private quantum channel using the fermionic unitaries only within cardinality $|\cl{U}|\ge2\kappa M$, beside $M^2$ in the optimal case. If we take a quantum purification, which describes that any mixed state can be transformed into a higher dimensional pure state, then we can always create a pure entangled state on fermionic Gaussian systems, for example in Ref.~\cite{TRCG19}.

\bigskip
\section{Conclusion}\label{sec:conclusion}
In this paper, we have firstly proposed an approximate private quantum channel for the fermionic Gaussian systems namely, $\varepsilon$-FPQC, and we make a useful formula to construct the quantum channel explicitly including its cardinality of needing unitary operations. While the fermionic PQC is needed exactly $4M^2$ fermionic unitary operations to encrypt an $M$-mode fermionic Gaussian state, our $\varepsilon$-FPQC is only sufficient to consume the number of unitary operations about $O(M\log M)$ in Proposition~\ref{prop2}. 

Beyond the bosonic Gaussian quantum channels, we expect that this kind of a research on fermionic channels will be meaningful for deep understanding of the quantum channel capacity problems i.e., the additivity violations for general capacities in broad Gaussian regimes. That is, if we know the exact form of a quantum purified state, which has a fermionic maximal entanglement, a research on the channel capacity problems could be also useful.

Finally, we remark on a relation between our approximate fermionic private quantum channel ($\varepsilon$-FPQC) and other information-theoretic settings. For examples, Clifford group $\cl{C}_n$ for $n$-qubit states is a subgroup of the unitary group $\cl{U}(d)$ for a qudit (that is, $d$-dimensional quantum state). Thus, we can easily observe that the above construction (i.e., $\varepsilon$-FPQC) has very similar structure to the $n$-qubit secure protocol for quantum sequential transmission~\cite{JK15}, magic-state construction~\cite{HJK+19}, and an operator system in mathematics for the error correction schemes in qubit levels~\cite{LNPST11}.

\section{Acknowledgments}
This work was supported by the National Research Foundation of Korea through a grant funded by the Ministry of Education (NRF-2018R1D1A1B07047512) and the Ministry of Science and ICT (NRF-2017R1E1A1A03070510).

%

\end{document}